\documentclass[a4paper,11pt]{article}
\usepackage{a4wide}
\usepackage{amsmath}
\usepackage{amsfonts}
\usepackage{amssymb}
\usepackage{amsthm}
\usepackage[utf8]{inputenc}
\usepackage{cellspace}
\usepackage[small,sc]{caption}
\usepackage{booktabs, multicol, multirow}
\usepackage{enumitem}

\theoremstyle{plain}
\newtheorem{theorem}{Theorem}
\newtheorem{restate}{Theorem}
\newtheorem{lemma}{Lemma}

\theoremstyle{definition}
\newtheorem{definition}{Definition}

\newcommand{\F}{\mathbb{F}}
\newcommand{\R}{\mathbb{R}}
\newcommand{\A}{\mathcal{A}}

\newcommand{\colvec}[1]{\begin{bmatrix}#1\end{bmatrix}} 
\newcommand{\mat}[1]{\mathbf{#1}}
\newcommand{\vect}[1]{\boldsymbol{#1}}
\DeclareMathOperator{\poly}{poly}
\DeclareMathOperator{\diag}{diag}

\setlist[1]{itemsep=-4pt}

\begin{document}

\title{Generating $k$-independent variables in constant time}

\author{
	Tobias Christiani\\
	\small \texttt{tobc@itu.dk}\\
	\small IT University of Copenhagen
	\and
	Rasmus Pagh\\
	\small \texttt{pagh@itu.dk}\\
	\small IT University of Copenhagen
\thanks{The research leading to these results has received funding from the European Research Council under the European Union's Seventh Framework Programme (FP7/2007-2013) / ERC grant agreement no. [614331].}
}

\renewcommand\footnotemark{}

\date{\vspace{-5ex}}

\maketitle

\begin{abstract}
The generation of pseudorandom elements over finite fields is fundamental to the time, space and randomness complexity of randomized algorithms and data structures. 
We consider the problem of generating $k$-independent random values over a finite field $\F$ in a word RAM model equipped with constant time addition and multiplication in $\F$, and present the first nontrivial construction of a generator that outputs each value in \emph{constant time}, not dependent on~$k$.
Our generator has period length $|\F|\poly \log k$ and uses $k \poly(\log k) \log |\F|$ bits of space, which is optimal up to a $\poly \log k$ factor.
We are able to bypass Siegel's lower bound on the time-space tradeoff for \mbox{$k$-independent} functions by a restriction to sequential evaluation.  
\end{abstract}

\section{Introduction}
Pseudorandom generators transform a short random seed into a longer output sequence. 
The output sequence has the property that it is indistinguishable from a truly random sequence by algorithms with limited computational resources. 
Pseudorandom generators can be classified according to the algorithms (distinguishers) that they are able to fool.
An algorithm from a class of algorithms that is fooled by a generator can have its randomness replaced by the output of the generator,
while maintaining the performance guarantees from the analysis based on the assumption of full randomness.
When truly random bits are costly to generate or supplying them in advance requires too much space, a pseudorandom generator can reduce the time, space and randomness complexity of an algorithm.

This paper presents an explicit construction of a pseudorandom generator that outputs a $k$-independent sequence of values in \emph{constant time} per value, not dependent on $k$, on a word RAM~\cite{hagerup1998}. 
The generator works over an arbitrary finite field that allows constant time addition and multiplication over $\F$ on the word RAM.

Previously, the most efficient methods for generating $k$-independent sequences were either based on multipoint evaluation of degree $k-1$ polynomials, or on direct evaluation of constant time hash functions.
Multipoint evaluation has a time complexity of $O(\log^{2} k \log \log k)$ field operations per value while hash functions with constant evaluation time use excessive space for non-constant $k$ by Siegel's lower bound~\cite{siegel2004}.
We are able to get the best of both worlds: constant time generation and near-optimal seed length and space usage. 

\paragraph{Significance.}
In the analysis of randomized algorithms and in the hashing literature in particular, $k$-independence has been the dominant framework for limited randomness. 
Sums of $k$-independent variables have their $j$th moment identical to fully random variables for $j \leq k$ which preserves many properties of full randomness.
For output length $n$, $\Theta(\log n)$-independence yields Chernoff-Hoeffding bounds~\cite{schmidt1995} and random graph properties~\cite{alon2008},
while $\Theta(\poly \log n)$-independence suffices to fool $AC^{0}$ circuits~\cite{braverman2010}.

Our generator is particularly well suited for randomized algorithms with time complexity $O(n)$ that use a sequence of $k$-independent variables of length $n$, for non-constant $k$.
For such algorithms, the generation of $k$-independent variables in constant time by evaluating a hash function over its domain requires space $O(n^{\epsilon})$ for some constant $\epsilon > 0$. 
In contrast, our generator uses space $O(k \poly \log k)$ to support constant time generation. 
Algorithms for randomized load balancing such as the simple process of randomly throwing $n$ balls into $n$ bins fit the above description and presents an application of our generator.
Using the bounds by Schmidt et al.~\mbox{\cite[Theorem 2]{schmidt1995}} it is easy to show that $\Theta(\log n / \log \log n)$-independence suffices to obtain a maximal load of any bin of $O(\log n / \log \log n)$ with high probability.
This guarantee on the maximal load is asymptotically the same as under full randomness. 
Using our generator, we can allocate each ball in constant time using space $O(\log n \poly \log \log n)$ compared to the lower bound of $O(n^{\epsilon})$ of hashing-based approaches to generating $k$-independence. 
In Section \ref{sec:loadbalancing} we show how our generator improves upon existing solutions to a dynamic load balancing problem.

The generation of pseudorandomness for Monte Carlo experiments presents another application.
Limited independence between Monte Carlo experiments can be shown to yield Chernoff-like bounds on the deviation of an estimator from its expected value. 
Consider a randomized algorithm~$\A(Y)$ that takes $m$ random elements from $\F$ encoded as a string $Y$ and returns a value in the interval $[0,1]$.    
Let $\mu_{\A} > 0$ denote the expectation of the value returned by~$\A(Y)$ under the assumption that $Y$ encodes a truly random input.
Define the estimator
\begin{equation}
\hat{\mu}_{\A} = \frac{1}{t}\sum_{i=1}^{t}\A(Y_{i}).
\end{equation}
Due to a result by Schmidt et al. \cite[Theorem 5]{schmidt1995}, for every choice of constants $\epsilon, \alpha > 0$,
it suffices that $Y_{1}, \dots, Y_{t}$ encodes a sequence of $\Theta(m \log t)$-independent variables over $\F$ to yield the following high probability bound on the deviation of $\hat{\mu}_{\A}$ from $\mu_{\A}$.
\begin{equation}
Pr[|\hat{\mu}_{\A} - \mu_{\A}| \geq \epsilon\mu_{\A}] \leq O(t^{-\alpha}). 
\end{equation}
We hope that our generator can be a useful tool to replace heuristic methods for generating pseudorandomness in applications where theoretical guarantees are important. 
In order to demonstrate the practicality of our techniques, we present experimental results on a variant of our generator in Section \ref{sec:experiments}. 
Our experiments show that $k$-independent values can be generated nearly as fast as output from heuristic pseudorandom generators, even for large $k$.

\paragraph{Methods.}
Our construction is a surprisingly simple combination of bipartite unique neigbor expanders with multipoint polynomial evaluation.
The basic, probabilistic construction of our generator proceeds in two steps: 
First we use multipoint evaluation to fill a table with \mbox{$\Theta(k)$-independent} values from a finite field, using an average of $\poly \log k$ operations per table entry.
Next we apply a bipartite unique neighbor expander with constant outdegree and with right side nodes corresponding to entries in the table and a left side that is $\poly \log k$ times larger than the right side.
For each node in the left side of the expander we generate a $k$-independent value by returning the sum of its neighboring table entries.
Our main result stated in Theorem 1 uses the same idea, but instead of relying on a single randomly constructed expander graph, 
we employ a cascade of explicit constant degree expanders and show that this is sufficient for constant time generation.   

\paragraph{Relation to the literature.}
Though the necessary ingredients have been known for around 10 years, we believe that a constant time generator has evaded discovery by residing in a blind spot between the fields of hashing and pseudorandom generators. 
The construction of constant time $k$-independent \emph{hash functions} has proven to be a difficult task, and a fundamental result by Siegel~\cite{siegel2004} showed a time-space tradeoff that require hashing-based generators with sequence length $n$ to use $O(n^{\epsilon})$ space for some constant $\epsilon > 0$. 
On the other hand, from the point of view of pseudorandom generators, a generator of $k$-independent variables, for non-constant~$k$, can not be used as an efficient method of derandomization: 
A lower bound by Chor et al.~\cite{chor1985} shows that the sample space of such generators must be superpolynomial in their output length. 
Consequently, research shifted towards generators that produce other types of outputs such as biased sequences or almost $k$-independent variables \cite{alon1992, naor1993, goldreich2010}. 

It is relevant to ask whether there already exist constructions of constant time pseudorandom generators on the word RAM that can be used instead of generators that output $k$-independent variables. 
For example, Nisan's pseudorandom generator~\cite{nisan1992} uses constant time to generate a pseudorandom word and has remarkably strong properties: 
Every algorithm running in $\textsc{space}(s)$ that uses $n$ random words can have its random input replaced by the output of a constant time generator with seed length $O(s \log n)$. 
The probability that the outcome of the algorithm differs when using pseudorandomness as opposed to statistical randomness is decreasing exponentially in the seed length. 

In spite of this strong result, there are many natural applications where the restrictions on Nisan's model means that we cannot use his generator directly to replace the use of a $k$-generator. 
An example is the analysis that uses a union bound over all subsets of $k$ words of a randomly generated structure described by $n$ words.
Algorithms shown to be derandomized by Nisan's generator are restricted to one-way access to the output of the generator. 
Therefore the output of Nisan's generator can not be used to derandomize an algorithm that tests for the events of the union bound without using excessive space.
In this case, \mbox{$k$-independence} can directly replace the use of full randomness without changing the analysis.
\subsection{Our contribution}
We present three improved constructions of \emph{$k$-generators}, formally defined in Section~\ref{sec:preliminaries}, that are able to generate a sequence of $k$-independent values over a finite field $\F$.
Our results are stated in a word RAM model equipped with constant time addition and multiplication in $\F$.     
Our main result is a fully explicit generator:
\begin{theorem}\label{thm:explicit}
For every finite field $\F$ with constant time arithmetic there exists a data structure that for every choice of $k \leq |\F| /\! \poly \log |\F|$ is an explicit constant time $k$-generator.
The generator has range $\F$, period $|\F|\poly \log k$, and seed length, space usage and initialization time $k \poly \log k$.
\end{theorem}
We further investigate how the space usage and seed length may be reduced by employing a probabilistic construction that has a certain probability of error:
\begin{theorem} \label{thm:existence}
For every finite field $\F$ with constant time arithmetic and every choice of positive constants $\varepsilon$, $\delta$ there exists a data structure 
that for every choice of $k = O(|\F|)$ is a constant time $k$-generator with failure probability $\delta$,
range~$\F$, period $|\F|$, seed length $O(k)$, space usage $O(k \log^{2+\varepsilon}k)$, and initialization time $O(k \poly \log k)$.
\end{theorem}
Finally, we improve existing $k$-generators with optimal space complexity:
\begin{theorem} \label{thm:fastmultipoint}
For every finite field $\F$ that supports computing the discrete Fourier transform of length $k$ in $O(k \log k)$ operations, 
there exists a data structure that, for every choice of $k \leq |\F|$ and given a primitive element $\omega$, 
is an explicit $O(\log k)$~time $k$-generator with range $\F$, period $|\F|$, seed length $k$, space usage $O(k)$, and initialization time $O(k \log k)$.
\end{theorem}

Table~\ref{tab:results} summarizes our results along with previous methods of generating sequences of $k$-independent values over $\F$. 
All the methods output sequences that have a length of at least~$|\F|$.

\begin{table}[htpb]
  \centering
  \small
    \begin{tabular}{lllll}
    \toprule
	{\bf Construction}                        & {\bf Time}               & {\bf Space}                     & {\bf Seed length}    & {\bf Comment}  \\ \midrule
    Polynomials \cite{joffe1974,wegman1981}   & $O(k)$                   & $O(k)$                          & $k$                  & \\
	Multipoint \cite{gathen2013}              & $O(\log^2 k \log\log k)$ & $O(k \log k)$                   & $k$			      & \\
	Multipoint \cite{bostan2005}              & $O(\log k \log\log k)$   & $O(k)$                          & $k$		          & Requires $\omega$. \\ 
	Siegel \cite{siegel2004}                  & $O(1)$                   & $O(|\F|^{\varepsilon})$         & $O(k)$               & Probabilistic. \\
	Theorem \ref{thm:explicit}                & $O(1)$                   & $k \poly \log k$                & $k \poly  \log k$    & Explicit.  \\
	Theorem \ref{thm:existence}               & $O(1)$                   & $O(k \log^{2 + \varepsilon}k)$  & $O(k)$               & Probabilistic. \\
	Theorem \ref{thm:fastmultipoint}          & $O(\log k)$              & $O(k)$                          & $k$
    & Requires $\omega_{k}$, FFT. \\
   \bottomrule
    \end{tabular}
\caption{
Overview of generators that produce a $k$-independent sequence over a finite field $\F$. 
We use $\varepsilon$ to denote an arbitrary positive constant and $\omega$ and $\omega_{k}$ to denote, respectively, a primitive element and a $k$-th root of unity of $\F$.
The unit for space and seed length is the number of elements of $\F$ that need to be stored, i.e., a factor $\log_2 |\F|$ from the number of bits. 
Probabilistic constructions rely on random generation of objects for which no explicit construction is known, and may fail with some probability.
}
\label{tab:results}
\end{table}

\paragraph{Overview of paper}
In Section \ref{sec:preliminaries} we define \mbox{$k$-generators} and related concepts and review results that lead up to our main results.
Section \ref{sec:explicit} presents the details of our explicit construction of constant time generators. 
In Section \ref{sec:probabilistic} we apply the same techniques with a probabilistic expander construction to obtain generators with improved space and randomness complexity. 
Section \ref{sec:faster} presents an algorithm for evaluating a polynomial over all elements of $\F$ that improves existing generators with optimal space.
Section \ref{sec:wordRAM} shows how arithmetic over $\F_{p}$ can be implemented in constant time on a standard word RAM with integer multiplication and also reviews algorithms and the state of hardware support for $\F_{2^{w}}$.
Section \ref{sec:loadbalancing} applies our generator to improve the time-space tradeoff of previous solutions to a load balancing problem.
Section \ref{sec:experiments} presents experimental results on the generation time of different $k$-generators for a range of values of $k$. 
\section{Preliminaries} \label{sec:preliminaries}
We begin by defining two fundamental concepts:
\begin{definition}
A sequence $(X_{1}, X_{2}, \dots, X_{n})$ of $n$ random variables with finite range $R$ is an \emph{$(n,k)$-sequence} if the variables at every set of $k$ positions in the sequence are independent and uniformly distributed over $R$. 
\end{definition}
\begin{definition}
A family of functions ${\mathcal{F} \subseteq \{f \mid f \colon U \to R \}}$ is \emph{$k$-independent} if for every set of $k$ distinct inputs $x_{1}, x_{2}, \dots, x_{k}$ 
it holds that $f(x_{1}), f(x_{2}), \dots, f(x_{k})$ are independent and uniformly distributed over $R$ when $f$ is selected uniformly at random from $\mathcal{F}$. 
We say that a function $f$ selected uniformly at random from $\mathcal{F}$ is a \emph{$k$-independent function}.
\end{definition}
We now give a formal definition of the generator data structure. 
\begin{definition}
A \emph{$k$-generator} with range $R$, period $n$ and failure probability $\delta$ is a data structure with the following properties:
\begin{itemize}
\item[--] It supports an initialization operation that takes a random seed $s$ as input.
\item[--] After initialization it supports an \texttt{emit()} operation that returns a value from $R$. 
\item[--] There exists a set $B$ such that $\Pr[s \in B] \leq \delta$ and conditioned on $s \not\in B$ the sequence $(X_{1}, X_{2}, \dots, X_{n})$ of values returned by \texttt{emit()} is an $(n,k)$-sequence. 
\end{itemize}
A $k$-generator is \emph{explicit} if the initialization and emit operation has time complexity $\poly k$ and the probability of failure is zero. 
We refer to a $k$-generator as a constant time $k$-generator if the \texttt{emit()} operation has time complexity $O(1)$, not dependent on $k$.
\end{definition}

A $k$-generator differs from a data structure for representing a $k$-independent hash function by only allowing sequential access to the underlying $(n,k)$-sequence. 
It is this restriction on generators that allows us to obtain a better time-space tradeoff for the problem of generating $k$-independent variables than is possible by using a $k$-independent hash function directly as a generator.  
We are interested in the following parameters of $k$-generators: seed length, period, probability of failure, space needed by the data structure, the time complexity of the initialization operation and the time complexity of a single \texttt{emit()} operation.

\paragraph{Model of computation.}
Our results are stated in the word RAM model of computation with word length ${w = \Theta(\log|\F|)}$ bits. 
In addition to the standard bit manipulation and integer arithmetic instructions, we also assume the ability to perform arithmetic operations $(+, -, \times)$ over $\F$ in constant time. 
In the context of our results that use abelian groups $(A, +)$ we assume that an element of $A$ can be stored in a constant number of words and that addition can be performed in constant time. 

Let $\F_{q}$ denote a field of cardinality $q = p^{z}$ for $p$ prime and $z$ a positive integer.
Constant time arithmetic in $\F_{p}$ is supported on a standard word RAM with integer multiplication \cite{granlund1994}. 
Section \ref{sec:wordRAM} presents additional details about the algorithms required to implement finite field arithmetic over $\F_p$ and $\F_{2^{w}}$ and how they relate to a standard word RAM with integer multiplication.
\subsection{$k$-independent functions from the literature} \label{sec:hashing}
We now review the literature on $k$-independent functions and how they can be used to construct $k$-generators.
We distinguish between a $k$-independent function $f : U \to R$ and a $k$-independent hash function by letting the latter refer to a data structure 
that after initialization supports random access to the $(n,k)$-sequence defined by evaluating $f$ over~$U$.  
There exists an extensive literature that focuses on how to construct $k$-independent hash functions that offer a favorable tradeoff between representation space and evaluation time \cite{dietzfelbinger2012}. 
We note that a family of $k$-independent hash functions can be used to construct a $k$-generator by setting the seed to a random function in the family.

\paragraph{Constant time $k$-independent hash functions.}
A fundamental cell probe lower bound by Siegel \cite{siegel2004} shows that a data structure to support constant time evaluation of $f$ on every input in $U$ 
cannot use less than $\Omega(|U|^{\epsilon})$ space for some constant $\epsilon > 0$. 
This bound holds even for amortized constant evaluation time over functions in the family and elements in the domain.
From Siegel's lower bound, it is clear that we cannot use $k$-independent hash functions directly to obtain a constant time $k$-generator that uses only $O(k \poly \log k)$ words of space.

Known constructions of $k$-independent hash functions with constant evaluation time are based on expander graphs.
Siegel \cite{siegel2004} gave a probabilistic construction of a family of \mbox{$k$-independent} hash functions in the word RAM model based on an iterated product of bipartite expander graphs. 
Thorup \cite{thorup2013} showed that a simple tabulation hash function with high probability yields the type of expander graphs required by Siegel's construction.
Unfortunately only randomized constructions of the expanders required by these hash functions is known, introducing a positive probability of error in \mbox{$k$-generators} based on them.

\paragraph{Polynomials.}
Here we briefly review the classic construction of $k$-independent functions based on polynomials over finite fields.  
\begin{lemma}[Joffe \cite{joffe1974}, Carter and Wegman \cite{wegman1981}] \label{lem:kpoly}
For every choice of finite field $\F$ and every $k \leq |\F|$, let $\mathcal{H}_{k} \subset \F[X]$ be the family of polynomials of degree at most $k-1$ over $\F$.
${\mathcal{H}_{k} \subset \{ f \mid f \colon \F \to \F \}}$ is a family of $k$-independent functions.
\end{lemma}

An advantage of using families of polynomials as hash functions is that they use near optimal randomness, allow any choice of $k \leq |\F|$, and have no probability of failure. 
It can also be noted that in the case where $k = O(\log |\F|)$ and we are restricted to linear space $O(k)$, 
polynomial hash functions evaluated using Horner's scheme are optimal \mbox{$k$-independent} hash functions \cite{larsen2012, siegel2004}.   

Using slightly more space and for sufficiently large $k$, a data structure by Kedlaya and Umans \cite{kedlaya2008} supports evaluation of a polynomial of degree $k$ over $\F$.
The space usage and preprocessing time of their data structure is $k^{1 + \epsilon}\log^{1 + o(1)}|\F|$ for constant $\epsilon > 0$.
After preprocessing a polynomial $f$, the data structure can evaluate $f$ in an arbitrary point of $\F$ using time $\poly(\log k)\log^{1 + o(1)}|\F|$.         

\paragraph{Multipoint evaluation.}
Using algorithms for multipoint evaluation of polynomials we are able to obtain a \mbox{$k$-generator} with $\poly \log k$ generation time and space usage that is linear in $k$. 
Multipoint evaluation of a polynomial~${f \in \F[X]}$ of degree at most $k-1$ in $k$ arbitrary points of~$\F$ has a time complexity of $O(k \log^{2} k \log \log k)$ in the word RAM model that supports field operations~\mbox{\cite[Corollary 10.8]{gathen2013}}. 
Bostan and Schost \cite{bostan2005} mention an algorithm for multipoint evaluation of $f$ over a geometric progression of $k$ elements with running time $O(k \log k \log \log k)$. 
In order to use this method to construct a $k$-generator with period $|\F|$ 
it is necessary to know a primitive element $\omega$ of $\F_{q}$ so we can perform multipoint evaluation over $\F^{*} = \{\omega^{0}, \omega^{1}, \dots, \omega^{q-2} \}$. 
Given the prime factorization of $q - 1$ there exists a Las Vegas algorithm for finding $\omega$ with expected running time $O(\log^{4} q)$~\mbox{\cite[Chapter 11]{shoup2009}}. 
In the following lemma we summarize the properties of $k$-generators based on multipoint evaluation of polynomials over finite fields.         

\begin{lemma}[{Gathen and Gerhard \cite[Corollary 10.8]{gathen2013}, Bostan and Schost \cite{bostan2005}}] \label{lem:multipoint}
For every finite field $\F$ there exists for every $k \leq |\F|$ and bijection $\pi : [|\F|] \to \F$ an explicit $k$-generator with period $|\F|$ and seed length $k$.
The space required by the generator and the initialization and generation time depends on the choice of $\pi$ and multipoint evaluation algorithm.
\begin{itemize}
\item[--] For arbitrary choice of $\pi$ there exists a $k$-generator with generation time $O(\log^{2} k \log \log k)$, intialization time $O(k \log^{2} k \log \log k)$ and space usage $O(k \log k)$. 
\item[--] Given a primitive element $\omega$ of $\F$ and a bijection $\pi(i) = \omega^{i}$ there exists a generator with generation time $O(\log k \log \log k)$, initialization time $O(k \log k \log \log k)$ and space usage $O(k)$.
\end{itemize}
\end{lemma}

\paragraph{Space lower bounds.}
Since randomness can be viewed as a resource like time and space, we are naturally interested in generators that can output long $k$-independent sequences using as few random bits as possible. 
Families of \mbox{$k$-independent functions} $f : U \rightarrow R$ with $U = R$ and $k \leq |U|$ will trivially have to use at least $k \log |U|$ random bits --- a bound matched by polynomial hash functions. 
We are often interested in generators with $|U| \gg |R|$, for example if we wish to use a generator for randomized load balancing in the heavily loaded case. 
A lower bound by Chor et al.~\cite{chor1985} shows that even in this case the minimal seed length required for $k$-independence is $\Omega(k \log |U|)$ for every $|R| \leq |U|$.
\subsection{Expander graphs}
All graphs in this paper are bipartite with $cm$ vertices on the left side, $m$ vertices on the right side and left outdegree $d$.
Graphs are specified by their edge function $\Gamma : [cm] \times [d] \to [m]$ where the notation $[n]$ is used to denote the set $\{0,1,\dots,n-1\}$.
Let $S$ be a subset of left side vertices. 
For convenience we use $\Gamma(S)$ to denote the neighbors of $S$. 
\begin{definition}
The bipartite graph $\Gamma : [cm] \times [d] \to [m]$ is \mbox{\emph{$(c,m,d,k)$-unique}} ($k$-unique) 
if for every $S \subseteq [cm]$ with $|S| \leq k$ there exists $y \in \Gamma(S)$ such that $y$ has a unique neighbor in $S$. 
An expander graph is \emph{explicit} if it has a deterministic description and $\Gamma$ is computable in time polynomial in $\log cm + \log d$. 
\end{definition}
The performance of our generator constructions are directly tied to the parameters of such expanders. 
In particular, we would like explicit expanders that simultanously have a low outdegree $d$, are highly unbalanced and are $k$-unique for $k$ as close to $m$ as possible.
A direct application of a result by Capalbo et al. \cite[Theorem 7.1]{capalbo2002} together with an equivalence relation between different types of expander graphs from Ta-Shma et al. \cite[Theorem 8.1]{tashma2007} yields explicit constructions of unbalanced unique neighbor expanders.\footnote{We state the results here without the restriction from \cite{capalbo2002} that $c$ and $m$ are powers of two. We do this to simplify notation and it only affects constant factors in our results.}  
\begin{lemma}[Capalbo et al. {\cite[Theorem 7.1]{capalbo2002}}] \label{lem:explicit}
For every choice of $c$ and $m$ there exists a $(c,m,d,k)$-unique expander with $d = \poly \log c$ and $k = \Omega(m/d)$. For constant $c$ the expander is explicit. 
\end{lemma}
We note the following simple technique for constructing a larger $k$-unique expander from a smaller $k$-unique expander.
\begin{lemma} \label{lem:stacking}
Let $\Gamma$ be a $(c,m,d,k)$-unique expander with $cm \times m$ adjacency matrix $\mat{M}$.
For any positive integer $b$ define $\Gamma^{(b)}$ as the bipartite graph with block diagonal adjacency matrix $\mat{M}^{(b)} = \diag(\mat{M}, \dots, \mat{M})$ with $b$ blocks in the diagonal.
Then $\Gamma^{(b)}$ is a $(c, bm, d, k)$-unique expander.
\end{lemma}

\paragraph{From expanders to independence.}
By associating each right vertex in a $(c,m,d,k)$-unique expander with a position in a $(m,dk)$-sequence over an abelian group $(A,+)$, we can generate a $(cm,k)$-sequence over $A$.
This approach was pioneered by Siegel and has been used in different constructions of families of $k$-independent hash functions~\cite{siegel2004, thorup2013}.    
\begin{lemma}[Siegel {\cite[Lemma 2.6, Corollary 2.11]{siegel2004}}] \label{lem:expanderhashing}
Let $\Gamma : [cm] \times [d] \to [m]$ be a $k$-unique expander and let $h : [m] \rightarrow A$ be a $dk$-independent function with range an abelian group. 
Let $g : [cm] \rightarrow A$ be defined as  
\begin{equation}
g(x) = \sum_{y \in \Gamma(\{x\})}h(y).
\end{equation}
Then $g$ is a $k$-independent function.
\end{lemma}
\section{Explicit constant time generators} \label{sec:explicit}
In this section we show how to obtain a constant time \mbox{$k$-generator} by combining an explicit $ \poly k$-generator with a cascading composition of unbalanced unique neighbor expanders. 
Our technique works by generating a small number of highly independent elements in an abelian group and then successively applying constant degree expanders to produce a greater number of less independent elements. 
We continue this process up until the point where the final number of elements is large enough to match the cost of generating the smaller batch of highly independent elements.    

The generator has two components.
The first component is an explicit $m$-generator $g_{0} : [n] \to A$ with period $n$ and range an abelian group $A$.
The second component is an explicit sequence $\left(\Gamma_{i}\right)^{t}_{i = 1}$ of unbalanced unique neighbor expanders.
The expanders are constructed such that the left side of the $i$th expander matches the right side of the $(i+1)$th expander.
By Lemma \ref{lem:explicit}, for every choice of imbalance $c$, target independence $k$ and length of the expander sequence $t$ there exists a sequence of expanders with the property that 
\begin{equation}
\Gamma_{i} \text{ is } (c, c^{i-1}m, d, d^{t-i}k)\text{-unique}, \label{eq:expandersequence}
\end{equation}
for $m = O(d^{t}k)$ and $d = \poly \log c$.
For constant $c$ each expander in the sequence is explicit.

We now combine the explicit $m$-generator $g_{0}$ and the sequence of expanders $\left(\Gamma_{i}\right)^{t}_{i = 1}$ to define the $k$-independent function $g_{t}$.
Let $b = m/n$ and assume for simplicity that $m$ divides $n$. 
For each $\Gamma_{i}$ we use the technique from Lemma \ref{lem:stacking} to construct a $(c, c^{i-1}n, d, d^{t-i}k)$-unique expander $\Gamma_{i}^{(b)}$.
Let $x_{i}$ denote a number in $[c^{i}n]$ corresponding to a vertex in the right side of $\Gamma_{i}^{(b)}$. 
We are now ready to give a recursive definition of $g_{i} : [c^{i}n] \to A$.
\begin{equation}
g_{i}(x_{i}) = 
\sum \limits_{x_{i-1} \in \Gamma_{i}^{(b)}(\{x_{i}\})} g_{i-1}(x_{i-1}), \quad 1 \leq i \leq t.
\label{eq:gexplicit}
\end{equation}
\begin{lemma}
$g_{i}$ is $d^{t-i}k$-independent. 
\end{lemma}
\begin{proof}
We proceed by induction on $i$. 
By definition, $g_{0} : [n] \to A$ is $d^{t}k$-independent.
Assume by induction that $g_{i} : [c^{i}n] \to A$ is $d^{t-i}k$-independent.
By definition $\Gamma_{i+1}^{(b)}$ is a $(c, c^{i}n, d, d^{t-(i+1)}k)$-unique expander.
Applying Lemma \ref{lem:expanderhashing} we have that $g_{i+1} : [c^{i+1}n] \to A$ is $d^{t-(i+1)}k$-independent.
\end{proof}

We will now show that $g_{t}$ supports fast sequential evaluation and prove that we can use $g_{t}$ to construct an explicit constant time $k$-generator from any explicit \mbox{$m$-generator}, for an appropriate choice of $m$.
Divide the domain of each $g_{i}$ evenly into $b = n/m$ batches of size $c^{i}m$ corresponding to each block of the adjacency matrix of $\Gamma_{i}$ used to construct $\Gamma_{i}^{(b)}$ and index the batches by $j \in [b]$. 
In order to evaluate $g_{i+1}$ over batch number $j$ it suffices to know $\Gamma_{i+1}$ and the values of $g_{i}$ over batch number $j$.
Fast sequential evaluation of $g_{t}$ is achieved in the following steps.
First we tabulate the sequence of expanders $\left(\Gamma_{i}\right)^{t}_{i = 1}$ such that $\Gamma_{i}(\{x_{i}\})$ can be read in $d$ operations.
Secondly, to evaluate $g_{t}$ over batch $j$, we begin by tabulating the output of $g_{0}$ over batch $j$ and then successively apply our tabulated expanders to produce tables for the output of $g_{1}, g_{2}, \dots, g_{t}$ over batch $j$.

Given tables for the sequence of expanders and assuming that the generator underlying $g_{0}$ has been initialized, we now consider the average number of operations used per output when performing batch-evaluation of $g_{t}$.
The number of values output is $c^{t}m$.
The cost of emitting $m$ values from $g_{0}$ is by definition at most $\poly(m)$.
The cost of producing tables for the output of $g_{1}, g_{2}, \dots, g_{t}$ for the current batch is given by $\sum_{i=1}^{t}dc^{i}m = O(dc^{t}m)$ for $c > 1$.
The average number of operations used per output when performing batch-evaluation of $g_{t}$ is therefore bounded from above by
\begin{equation}
\frac{O(dc^{t}m) + \poly m}{c^{t}m} = O(d) + \frac{\poly m}{c^{t}}. \label{eq:averagetime}
\end{equation}
The following lemma states that we can obtain a constant time $k$-generator from every explicit $m$-generator by setting $t = O(\log k)$ and choosing $c$ to be an appropriately large constant.  
\begin{lemma} \label{lem:general}
Let $A$ be an abelian group with constant time addition. 
Suppose there exists an explicit $m$-generator with range $A$, period $n$ and space usage $\poly m$.
Then there exists a positive constant $\epsilon$ such that for every $k \leq m^{\epsilon}$ 
there exists an explicit constant time $k$-generator with range $A$, period $n$, and seed length, space usage and initialization time $\poly k$.
\end{lemma}
\begin{proof}
The sequence of expanders $\left(\Gamma_{i}\right)^{t}_{i = 1}$ with the properties given in \eqref{eq:expandersequence} exists for $m = O(d^{t}k)$ and $d = \poly \log c$ and is explicit for $c$ constant.
By inserting $m = O(d^{t}k)$ into equation \eqref{eq:averagetime} it can be seen that the average number of operations is constant for $c = O(1)$ and $t = O(\log k)$ with constants that depend on the parameters of the $m$-generator. 
The $k$-generator is initialized by initializing the $m$-generator, finding and tabulating the sequence of expanders and producing the first batch of values, all of which can be done in $\poly k$ time and space.
After initialization, each call to \texttt{emit()} will return a value from the current batch and use a constant number of operations for the task of preparing the next batch of outputs.    
\end{proof}

We now show our main theorem about explicit constant time $k$-generators over finite fields. 
The construction uses an explicit $m$-generator based on multipoint evaluation. 
Combined with the approach of Lemma \ref{lem:general} this yields a near-optimal time-space tradeoff for $k$-generation.
\begin{restate}[Repeated]
For every finite field $\F$ with constant time arithmetic there exists a data structure that for every choice of $k \leq |\F| /\! \poly \log |\F|$ is an explicit constant time $k$-generator.
The generator has range $\F$, period $|\F|\poly \log k$, and seed length, space usage and initialization time $k \poly \log k$.
\end{restate}
\begin{proof}
Fix the choice of finite field $\F$. 
By Lemma~\ref{lem:multipoint} there exists an explicit $m$-generator in $\F$ for $m \leq |\F|$ with period $|\F|$ that uses time $O(m \log^{3} m)$ to emit $m$ values.
Fix some constant $c > 1$ and let $\left(\Gamma_{i}\right)^{t}_{i = 1}$ denote an explicit sequence of constant degree expanders with the properties given by \eqref{eq:expandersequence}.
The average number of operations per $k$-independent value output by $g_{t}$ when performing batch evaluation is given by
\begin{equation}
\frac{O(dc^{t}m) + O(m \log^{3} m)}{c^{t}m} = O(d) + \frac{O(\log^{3} d^{t}k)}{c^{t}}. \label{eq:averagetimefield}
\end{equation}
Setting $t = O(\log \log k)$ and following the approach of Lemma \ref{lem:general} we obtain a $k$-generator with the stated properties. 
\end{proof}
Based on the discussion in a paper by Capalbo \cite{capalbo2005} that introduces unbalanced unique neighbor expanders for concrete values of $c$ and $d$, it appears likely that the constants hidden in Theorem \ref{thm:explicit} for the current best explicit constructions make our explicit generators unsuited for practical use since $c$ is close to $1$ when $d$ is reasonably small. 
The next section explores how randomly generated unique neighbor expanders can be used to show stronger existence results and yield $k$-generators with tractable constants.  
%
\section{Constant time generators with optimal seed length} \label{sec:probabilistic}
Randomly constructed expanders of the type used in this paper have stronger properties than known explicit constructions, and can be generated with an overwhelming probability of success.
There is no known efficient algorithm for verifying whether a given graph is a unique neighbor expander.
Therefore randomly generated expanders cannot be used to replace explicit constructions without some probability of failure.

In this section we apply the probabilistic method to show the existence of $k$-generators with better performance characteristics than those based on known explicit constructions of expanders. 
We are able to show the existence of constant time generators with optimal seed length that use $O(k\log^{2+\varepsilon}k)$ words of space for any constant $\varepsilon > 0$. 
Furthermore, such generators can be constructed for any choice of constant failure probability $\delta > 0$. 
The generators we consider in this section use only a single expander graph but are otherwise identical to the generators described in Section \ref{sec:explicit}.
Using a single expander graph suffices for constant time generation because the probabilistic constructions are powerful enough to support an imbalance of $c = \poly \log k$ while maintaining constant degree.
This imbalance is enough to amortize the cost of multipoint evaluation in a single expansion step as opposed to the sequence of explicit expanders employed in Theorem \ref{thm:explicit}. 
Our arguments are a straightforward application of the probabilistic method, but we include them for completeness and because we are interested in somewhat nonstandard parameters.

We consider the following randomized construction of a $(c,m,d,k)$-unique expander $\Gamma$. 
For each vertex $x$ in $[cm]$, we add an edge between $x$ and each distinct node of $d$ nodes selected uniformly at random from $[m]$.  
By a standard argument, the graph can only fail to be unique neighbor expander if there exists a subset $S$ of left hand side vertices with $|S| \leq k$ such that $|\Gamma(S)| \leq \lfloor d|S|/2 \rfloor$ \cite[Lemma 2.8]{siegel2004}.
In the following we assume that $kd \leq m$. 
\begin{align}
&\Pr[\Gamma \text{ is not a unique neighbor expander}] \notag \\
&\leq \Pr[\exists S \subseteq [cm], |S| \leq k : |\Gamma(S)| \leq \lfloor d|S|/2 \rfloor] \notag \\
&\leq \sum_{\substack{S \subseteq [cm]\\ |S| \leq k}} \Pr[|\Gamma(S)| \leq \lfloor d|S|/2 \rfloor] \notag \\
&\leq \sum_{i = 1}^{k} \binom{cm}{i} \binom{m}{\lfloor id/2 \rfloor} \left(\frac{\lfloor id/2 \rfloor}{m}\right)^{id} \notag \\ 
&\leq \sum_{i=1}^{k} \left(\frac{cme}{i}\right)^{i} \left(\frac{me}{id/2}\right)^{id/2} \left(\frac{id/2}{m}\right)^{id} \notag \\
&= \sum_{i=1}^{k} \left( cm e^{1 + d/2} \left(\frac{(d/2)i^{1 - 1/(d/2)}}{m}\right)^{d/2} \right)^{i} \label{eq:probexpander}
\end{align}
If the expression in the outer parentheses in \eqref{eq:probexpander} can be bounded from above by $1/2$ for $i = 1,2,\dots,k$, then the expander exists.
We also note that the randomized expander construction can be performed using $dk$-independent variables without changing the result in~\eqref{eq:probexpander}. 
Let $\gamma > 1$ be a number that may depend on $k$ and let $\delta$ denote an upper bound on the probability that the randomized construction fails. 
By setting $m = O(dk\gamma)$ we are able to obtain the following expression for the relation between $\delta$, the imbalance $c$ and the left outdegree bound $d$.
\begin{equation}
\delta = e\frac{cd}{\gamma^{d/2 - 1}} \label{eq:expanderparameters}
\end{equation}
Equation \eqref{eq:expanderparameters} reveals tradeoffs for the parameters of the randomly constructed $k$-unique expander graphs.
For example, increasing $\gamma$ makes it possible to make the graph more unbalanced while maintaining the same upper bound on the probability of failure $\delta$. 
The increased imbalance comes at the cost of an increase in $m$, the size of the right side of the graph. 
Similarly it can be seen how increasing $d$ can be used to reduce the probability of error.
Setting the parameters to minimize the space occupied by the expander while maintaining constant outdegree and by extension constant generation time, we obtain Theorem~\ref{thm:existence}.  
\begin{restate}[Repeated]
For every finite field $\F$ with constant time arithmetic and every choice of positive constants $\varepsilon$, $\delta$ there exists a data structure 
that for every choice of $k = O(|\F|)$ is a constant time $k$-generator with failure probability $\delta$,
range~$\F$, period $|\F|$, seed length $O(k)$, space usage $O(k \log^{2+\varepsilon}k)$, and initialization time $O(k \poly \log k)$.
\end{restate}
\begin{proof}
Let $\tilde{\varepsilon} < \varepsilon$ be a constant and set $\gamma = \log^{\tilde{\varepsilon}}k$.
Choosing $d$ to be a sufficiently large constant (dependent on $\tilde{\varepsilon}$), equation \eqref{eq:expanderparameters} shows that 
for every $\delta > 0$ there exists a $(c, m, d, k)$-unique expander $\Gamma$ with $c = \Omega(\log^{2+\varepsilon}k)$ and $m = O(k\gamma)$.
Using multipoint evaluation, the right side vertices of $\Gamma$ can be associated with $\Theta(k)$-independent variables over $\F$ using $O(k \log^{2 + \varepsilon}k)$ operations.
By the properties of $\Gamma$ and applying Lemma \ref{lem:expanderhashing} we are able to generate batches of $k$-independent variables of size $\Omega(k\log^{2+\varepsilon}k)$ using $O(k \log^{2+\varepsilon})$ operations.
The seed length of $O(k)$ holds by the observation that randomized construction of the expander only requires $O(k)$-independence.
The $O(k \poly \log k)$ initialization time is obtained by using multipoint evaluation to construct a table for $\Gamma$.
\end{proof}
\section{Faster multipoint evaluation for $k$-generators} \label{sec:faster}
This section presents an improved generator based directly on multipoint evaluation of a polynomial hash function $h \in \mathcal{H}_{k}$ over a finite field.
For our purpose of generating an $(n,k)$-sequence from $h$, we are free to choose the order of elements of $\F$ in which to evaluate $h$.  
We present an algorithm for the systematic evaluation of $h$ over disjoint size $k$ subsets of $\F$ using Fast Fourier Transform (FFT) algorithms.
Our technique yields a $k$-generator over $\F$ with generation time $O(\log k)$, and space usage and seed length that is optimal up to constant factors.
The algorithm depends upon the structure of $\F$, similarly to other FFT algorithms over finite fields \cite{bhattacharya2004}. 

The nonzero elements of $\F$ form a multiplicative cyclic group $\F^{*}$ of order $q-1$. 
The multiplicative group has a primitive element $\omega$ which generates $\F^{*}$.
\begin{equation}
\F^{*} = \{ \omega^{0}, \omega^{1}, \omega^{2}, \dots, \omega^{q-2} \}.
\end{equation}
For $k$ that divides $q-1$, we can construct a multiplicative subgroup $S_{k,0}^{*}$ of order $k$ with $\omega_{k} = \omega^{(q-1)/k}$ as the generating element. 
$S_{k,0}^{*}$ contains $k$ distinct elements of $\F$. 
Define for $j = 0,1, \dots, (q-1)/k - 1$, 
\begin{equation}
S_{k,j}^{*} = \omega^{j}S_{k,0} = \{ \omega^{j}\omega_{k}^{0}, \omega^{j}\omega_{k}^{1}, \dots, \omega^{j}\omega_{k}^{k-1} \}.  
\end{equation}
Viewed as subsets of $\F^{*}$ the sets $S_{k,j}^{*}$ form an exact cover of $\F^{*}$. 
We now consider how to evaluate a degree $k-1$ polynomial $h(x) \in \F[X]$ in the points of $S_{k,j}^{*}$. The polynomial takes the form
\begin{equation}
h(x) = a_{0}x^{0} + a_{1}x^{1} + \dots + a_{k-1}x^{k-1}.
\end{equation}
Rewriting the polynomial evaluation over $S_{k,j}^{*}$ in matrix notation:
\begin{equation}
\colvec{ h(\omega^{j}\omega_{k}^{0}) \\  h(\omega^{j}\omega_{k}^{1}) \\ h(\omega^{j}\omega_{k}^{2}) \\ \vdots \\ h(\omega^{j}\omega_{k}^{k-1})} = 
\colvec{
	\omega_{k}^{0\cdot0} & \omega_{k}^{0\cdot1} & \dots & \omega_{k}^{0\cdot(k-1)} \\
	\omega_{k}^{1\cdot0} & \omega_{k}^{1\cdot1} & \dots & \omega_{k}^{1\cdot(k-1)} \\
	\omega_{k}^{2\cdot0} & \omega_{k}^{2\cdot1} & \dots & \omega_{k}^{2\cdot(k-1)} \\
	\vdots & \vdots &  & \vdots \\
	\omega_{k}^{(k-1)\cdot0} & \omega_{k}^{(k-1)\cdot1} & \dots & \omega_{k}^{(k-1)\cdot(k-1)}
}
\colvec{\omega^{j \cdot 0}a_{0} \\ \omega^{j \cdot 1}a_{1} \\ \omega^{j \cdot 2}a_{2} \\ \vdots \\ \omega^{j \cdot (k-1)}a_{k-1}}
\label{eq:polymatrix}
\end{equation}
We assume that the coefficients of $h$ and $\omega^{j}$ are given and consider algorithms for efficient evaluation of the matrix-vector product.
The coefficients $\tilde{a}_{j,i} = \omega^{j \cdot i}a_{i}$ for $i = 0,1,\dots,k-1$ can be found in $O(k)$ operations and define a polynomial $\tilde{h}_{j}(x) = \sum_{i=0}^{k-1}\tilde{a}_{i,j}x^{i}$. 
Evaluating $\tilde{h}_{0}(x)$ over $S_{k,0}^{*}$ corresponds to computing the Discrete Fourier Transform over a finite field.
\begin{restate}[Repeated] 
For every finite field $\F$ that supports computing the discrete Fourier transform of length $k$ in~$O(k \log k)$ operations, 
there exists a data structure that, for every choice of $k \leq |\F|$ and given a primitive element $\omega$, 
is an explicit $O(\log k)$~time $k$-generator with range $\F$, period~$|\F|$, seed length $k$, space usage $O(k)$, and initialization time $O(k \log k)$.
\end{restate}
\begin{proof}
Evaluation of $\tilde{h}_{j}(x)$ over $S_{k,j}^{*}$ takes $O(k \log k)$ operations by assumption. 
For every batch $j$ starting at $j = 0$, the value of $\omega^{j}$ is stored and used to compute the coefficients of $\tilde{h}_{j+1}(x)$ ing $O(k)$ operations.
\end{proof}
We now discuss the validity of the assumption that we are able to compute the DFT over a finite field in $O(k \log k)$ operations.
Assume that $k \mid (q - 1)$ and that $\omega_{k}$ is known.
If $k$ is highly composite there exist Fast Fourier Transforms for computing \eqref{eq:polymatrix} in $O(k \log k)$ field operations~\cite{duhamel1990}.
If $k$ is not highly composite there exists an algorithm for computing the DFT in equation \eqref{eq:polymatrix} in $O(kz \log kz )$ operations for fields of cardinality $q = p^{z}$ in our model of computation~\cite{preparata1977}. 
For $q = p^{O(1)}$ this reduces to the desired $O(k \log k)$ operations.
\section{Finite field arithmetic on the word RAM} \label{sec:wordRAM}
Throughout the paper we have used as our model of computation a modified word RAM with constant time arithmetic $(+,-, \times)$ over a finite field $\F$.
In this section we show how our model relates to the more standard $\emph{multiplication model}$ defined as a word RAM with constant time arithmetic $(+, -, \times)$ over the integers $[2^{w}]$ for $w$-bit words \cite{hagerup1998}.

Arithmetic over $\F_{p}$ for prime $p$ is integer arithmetic modulo $p$. 
We now argue that arithmetic operations over $\F_{p}$ can be performed in $O(1)$ operations in the multiplication model.
Every integer $x$ can be written on the form $x = qp + r$ for non-negative integers $q, r$ with $r < p$.
Assume that $x$ can be represented in a constant number of words.
The problem of computing $r = x \bmod p$ can be solved by an integer division and $O(1)$ operations in the multiplication model due to the identity $r = x -  \lfloor x/p \rfloor p$.
An algorithm by Granlund and Montgomery \cite{granlund1994} computes $\lfloor x/p \rfloor$ for any constant $p$ using $O(1)$ operations in the multiplication model which gives the desired result.

Another finite field of interest is $\F_{2^{w}}$ due to the correspondence between field elements and bit vectors of length $w$.
We will argue that a word RAM model that supports constant time multiplication over $\F_{2^{w}}$ is not unrealistic considering current hardware.
Addition in $\F_{2^{w}}$ has direct support in standard CPU instruction sets through the XOR operation.
A multiplication of two elements $x$ and $y$ in $\F_{2^{w}}$ can be viewed as a two-step process.
First, we perform a carryless multiplication $z = x \cdot y$ of the representation of $x$ and $y$ as polynomials in $F_{2}[X]$. 
Second, we use a modular reduction to bring the product $x \cdot y$ back into $\F_{2^{w}}$, similarly to modular arithmetic over $\F_{p}$. 
Recently, hardware manufacturers have included partial support for multiplication in $\F_{2^{w}}$ with the CLMUL instruction for carryless multiplication \cite{gueron2014}. 
The modular reduction step is performed by dividing $x \cdot y$ by an irreducible polynomial $g$ and returning the remainder.
Irreducible polynomials $g$ that can be represented as sparse binary vectors with constant weight results in a constant time algorithm for modular reduction as presented by Gueron and Kounavis \cite{gueron2014}.
We briefly introduce the computation underlying the algorithm to show that its complexity depends on the number of {\tt 1}s in the binary representation of $g$.
Let $L^{w}$ and $M^{w}$ be functions that return the $w$ least, respectively most, significant bits of their argument as represented in $\F_{2^{2w}}$.  
The complexity of Gueron and Kounavis' algorithm for modular reduction of $z = x \cdot y$ is determined by the complexity of evaluating the expression
\begin{equation}
L^{w}(L^{w}(g)\cdot M^{w}(M^{w}(z) \cdot g)). \label{eq:clmul}
\end{equation}
Evaluating $L^{w}$ and $M^{w}$ is standard bit manipulation. 
For $g$ of constant weight, the carryless multiplications denoted by $\cdot$ in equation $\eqref{eq:clmul}$ can be implemented as a constant number of bit shifts and XORs. 
For every $w \leq 10000$ an irreducible trinomial or pentanomial ($g$ of weight at most 5) has been found~\cite{seroussi1998}.
Together with the hardware support for convolutions this allows us to implement fast multiplication over fields of practical interest. 
\section{A load balancing application} \label{sec:loadbalancing}
We next consider how our new generator yields stronger guarantees for load balancing.
Our setting is motivated by applications such as splitting a set of tasks of unknown duration among a set of $m$ machines, in order to keep the load as balanced as possible.
Once a task is assigned to a machine, it cannot be reassigned, i.e., we do not allow \emph{migration}.
For simplicity we consider the \emph{unweighted} case where we strive to keep the \emph{number} of tasks on each machine low, and we assume that $m$ divides $|\F|$ for some field $\F$ with constant time operations on a word RAM.
Suppose that each machine has capacity (e.g.~memory enough) to handle $b$ tasks at once, and that we are given a sequence of $t$ tasks $T_1,\dots,T_t$, where we identify each task with its duration (an interval in $\R$).
Now let $k = mb$ and suppose that we use our constant time $k$-generator to determine for each $i=1,\dots,t$ which machine should handle $T_i$.
(We emphasize that this is done without knowledge of $T_i$, and without coordination with the machines.)
Compared to using a fully random choice this has the advantage of requiring only $k\poly\log k$ words of random bits, which in turn may make the algorithm faster if random number generation is a bottleneck.
Yet, we are able to get essentially the same guarantee on load balancing as in the fully random case.
To see this let $L(x) = \{ i \; | \; x\in T_i\}$ be the set of tasks active at time $x$, and let $L_q(x)$ be the subset of $L(x)$ assigned to machine $q$ using our generator.
We have:
\begin{lemma}\label{lem:error}
For $\varepsilon > 0$, if $|L(x)| (1+\varepsilon) < mb$ then
$\Pr[\max_q |L_q(x)| > b] < m \exp(-\varepsilon^2 b / 3)$.
\end{lemma}
\begin{proof}
Since $|L(x)| < mb = k$ we have that the assignment of tasks in $L(x)$ to machines is uniformly random and independent.
This means that the number of tasks assigned to each machine follows a binomial distribution with mean $b/(1+\varepsilon)$, and we can apply a Chernoff bound of $\exp(-\varepsilon^2 b / 3)$ on the probability that more than $b$ tasks are assigned to a particular machine.
A union bound over all $m$ machines yields the result.
\end{proof}

Lemma~\ref{lem:error} allows us to give a strong guarantee on the probability of exceeding the capacity $b$ of a machine at any time, assuming that the average load is bounded by $b/(1+\varepsilon)$.
In particular, let $S\subseteq\R$ be a set of size at most $2t$ such that every workload $L(y)$ is equal to $L(x)$ for some $x\in S$.
The existence of $S$ is guaranteed since the $t$ tasks are intervals, and they have at most $2t$ end points.
This means that
$$\sup_{x\in\R} \max_q |L_q(x)| = \max_{x\in S} \max_q |L_q(x)|,$$
so a union bound over $x\in S$ gives
$$\Pr[\sup_{x\in\R} \max_q |L_q(x)| > b] < 2tm \exp(-\varepsilon^2 b / 3) \enspace .$$

For constant $\varepsilon$ and whenever $b = \omega(\log k)$ and $tm = 2^{o(b)}$ we get an error probability that is exponentially small in $b$.
Such a strong error guarantee can not be achieved with known constant time hashing methods~\cite{siegel2004,pagh2008,dietzfelbinger2003,thorup2013} in reasonable space, since they all have an error probability that decreases polynomially with space usage.
Even if explicit constructions for the expanders needed in Siegel's hash functions were found, the resulting space usage would be polynomially higher than with our $k$-generator.
\section{Experiments} \label{sec:experiments}
This section contains experimental results of an implementation of a $k$-generator over $\F_{2^{64}}$. 
There are two main components to the generator: an algorithm for filling a table of size $m$ with $dk$-independent variables and a bipartite unbalanced expander graph.

For the first component, we use an implementation of Gao-Mateer's additive FFT \cite[Algorithm 2.]{gao2010}.
Utilizing the Gao-Mateer algorithm we can generate a batch of $k$ elements of an $(|\F|, k)$-sequence using space $O(k)$ and $O(k \log^{2} k)$ operations on a word RAM that supports arithmetic over $\F$.
The additive complexity of the FFT algorithm is $O(k \log^{2} k)$ while the multiplicative complexity is $O(k \log k)$.
Addition in $\F_{2^{64}}$ is implemented as an XOR-operation on 64-bit words. 
Multiplication is implemented using the PCLMUL instruction along with the techniques for modular reduction by Gueron et
al. \cite{gueron2014} outlined in Section \ref{sec:wordRAM}.

For the second component we introduce a slightly different type of expander graphs that only work in the special case of
fields of characteristic two. 
Let $\F_{2^{w}}$ be a field of characteristic two and let $\mat{M}$ be a $cm \times m$ adjacency matrix of a graph $\Gamma$ where each entry of $\mat{M}$ is viewed as an element of $\F_{2^{w}}$.    
By a similar argument to the one used in Lemma \ref{lem:expanderhashing} the linear system $\mat{M}\vect{x}$ defines a $(cm, k)$-sequence if $\vect{x}$ is a vector of \mbox{$dk$-independent} variables over $\F_{2^{w}}$ and $\mat{M}$ has row rank at least $k$.
We consider randomized constructions of $\mat{M}$ over $\F_{2}$ with at most $d$ \texttt{1}s in each row and row rank at least $k$.
It is easy to see that a matrix $\mat{M}$ over $\F_{2}$ with these properties also defines a matrix with the same properties over $\F_{2^{w}}$.
Since $k$-uniqueness of $\Gamma$ implies that $\mat{M}$ has row rank $k$, but not the other way around, 
we are able to obtain better performance characteristics of generators over $\F_ {2^{w}}$ by focusing on randomized constructions of $\mat{M}$.

The matrix $\mat{M}$ is constructed in the following way.
Independently, for each $i \in [cm]$ sample $d$ integers uniformly with replacement from $[m]$ and define the $i$th row of $\mat{M}$ as the vector constructed by taking the zero vector and adding~\texttt{1}s in the $d$ positions sampled for row $i$.
Observe that if $\mat{M}$ does not have row rank at least $k$ then some non-empty subset of at most $k$ rows of $\mat{M}$ sum to the zero vector.
In order for a non-empty set of vectors over $\F_{2}^{m}$ to sum to the zero vector, the bit-parity must be even in each of the $m$ positions of the sum. 
The sum of any $i$ rows of $\mat{M}$ corresponds to a balls and bins process that distributes $id$ balls into $m$ bins, independently and uniformly at random.
Let $id$ be an even number. Then there are $(id - 1)!!$ ways of ordering the balls into pairs and the probability that the outcome is equal to any particular pairing is $(1/m)^{id/2}$. 
This yields the following upper bound on the probability that a subset of $i$ rows sums to zero: 
\begin{equation}
\beta_{pair}(i, d, m) = (id-1)!!\! \left(\frac{1}{m}\right)^{id/2}. \label{eq:libound} 
\end{equation}
A comparison between this bound and the bound for $k$-uniqueness from equation \eqref{eq:probexpander} shows that, for each term in the sum, 
the multiplicative factor applied to the binomial coefficient $\binom{cm}{i}$ is exponentially smaller in $id$ for the bound in \eqref{eq:libound}.

The pair-based approach which yields the bound $\beta_{pair}$ overestimates the probability of failure on subsets of size $i$, increasingly as $id$ grows large compared to $m$. 
We therefore introduce a different bound based on the Poisson approximation to the binomial distribution: 
the number of balls in each of the $m$ positions can approximately be modelled as independent Poisson distributed variables \cite[Ch. 5.4]{mitzenmacher2005}. 
The probability that that the parity is even in each of the $m$ positions in a sum of $i$ rows is bounded by
\begin{equation}
\beta_{poisson}(i, d, m) = e\sqrt{id}\left(\frac{1+e^{-2\frac{id}{m}}}{2}\right)^{m}, \label{eq:poibound} 
\end{equation}
where we use the same approach as Mitzenmacher et al.~\cite{mitzenmacher2014}.
For any given subset of rows of $\mat{M}$, we are free to choose between the two bounds. 
The probability that a randomly constructed matrix $\mat{M}$ fails to have rank at least $k$ can be bounded from above using a union bound over subsets of rows of $\mat{M}$.
\begin{equation}
\delta \leq \sum_{i=1}^{k}\binom{cm}{i}\min(\beta_{pair}(i, d, m), \beta_{poisson}(i, d, m)). \label{eq:combinedbound}
\end{equation}

We now consider the generation time of our implementation. 
Let $FFT_{dk}$ denote the time taken by the FFT algorithm to generate a $dk$-independent value and let $RA_{d,m}$ denote the time it takes to perform $d$ random accesses in a table of size $m$.  
The time taken to generate a value by the implementation of our generator is then given by
\begin{equation}
T = \frac{FFT_{dk}}{c} + RA_{d,m}. \label{eq:generationtime}
\end{equation}
In our experiments, the choice of parameters for the expander graphs were based on a search for the fastest generation time over every combination of imbalance $c \in \{16, 32, 64 \}$ and outdegree $d \in \{ 4, 8, 16 \}$.
Given choices of $d$, $c$ and independence $k$, the size of the right side of the expander $m$ was increased until existence could be guaranteed by the bound in \eqref{eq:combinedbound}.
The generator in the experiments had the restriction that $m \leq 2^{26}$ and we have measured $RA_{d,m}$ assuming that
the expander is read sequentially from RAM. 
The experiments were run on a machine with an Intel Core i5-4570 processor with 6MB cache and 8GB of RAM.

Table \ref{tab:experimentalresults} shows the generation time in nanoseconds per 64-bit output using Horner's scheme, Gao-Mateer's FFT and the implementation of our generator (FFT + Expander).
For the implementation of the generator, we also show the parameters of the randomly generated expander that yielded the fastest generation time among expanders in the search space.

The generation time for Horner's scheme is approximately linear in $k$ and logarithmic in $k$ for the FFT, as predicted by theory. 
The FFT is faster than using Horner's scheme already at $k = 64$ and orders of magnitude faster for large $k$.
Using our implementation of Gao-Mateer's FFT algorithm we are able to evaluate a polynomial of degree $2^{20}-1$ in $2^{20}$ points in less than a second. 
The same task takes over an hour when using Horner's rule, even with both algorithms using the same underlying implementation of algebraic operations in the field.

For small values of $k$, our generator is an order of magnitude faster than the FFT and comes close to the performance of the 64-bit C++11 implementation of the popular Mersenne Twister.
Our generator uses 25 nanoseconds to output a 1024-independent value. This is equivalent to an output of over 300MB/s.
The Mersenne Twister uses around 4 nanoseconds to generate a 64-bit value.

In practice, the memory hierarchy appears to be the primary obstacle to maintaining a constant generation time as $k$ increases.
Our generator reads the expander graphs sequentially and performs random lookups into the table of $dk$-independent values.
As $k$ grows large, the table can no longer fit into cache and for large imbalance $c$, the expander can no longer be stored in main memory.
Searching a wider range of expander parameters could easily yield a faster generation time, potentially at the cost of a larger imbalance $c$ or higher probability of failure $\delta$.

\begin{table}[htpb]
	\centering
	\begin{tabular}{rrr|rrrrr}
	
	\toprule
    
\multirow{2}{*}{$k$} & \multicolumn{1}{c}{Horner} & \multicolumn{1}{c|}{FFT}  & \multicolumn{5}{c}{FFT + Expander}       \\ 
           
& \multicolumn{1}{c}{ns} &   \multicolumn{1}{c|}{ns}  & \multicolumn{1}{c}{$c$} &  \multicolumn{1}{c}{$m$}    & \multicolumn{1}{c}{$d$} & \multicolumn{1}{c}{$\delta$}    & \multicolumn{1}{c}{ns}    \\ \midrule 
$2^{5}$    &       177 &  243  & 64 & $2^{13}$ &  8 & $10^{-7}$   &  15   \\    
$2^{6}$    &       361 &  294  & 64 & $2^{14}$ &  8 & $10^{-8}$   &  16   \\ 
$2^{7}$    &       730 &  338  & 64 & $2^{15}$ &  8 & $10^{-9}$   &  19   \\ 
$2^{8}$    &      1470 &  375  & 64 & $2^{16}$ &  8 & $10^{-10}$  &  23   \\ 
$2^{9}$    &      2950 &  412  & 64 & $2^{17}$ &  8 & $10^{-11}$  &  24   \\ 
$2^{10}$   &      5902 &  449  & 64 & $2^{18}$ &  8 & $10^{-12}$  &  25   \\ 
$2^{11}$   &     11808 &  487  & 32 & $2^{18}$ &  8 & $10^{-12}$  &  35   \\ 
$2^{12}$   &     23627 &  523  & 64 & $2^{18}$ & 16 & $10^{-29}$  &  43   \\ 
$2^{13}$   &     47183 &  561  & 32 & $2^{18}$ & 16 & $10^{-29}$  &  54   \\ 
$2^{14}$   &     94429 &  599  & 64 & $2^{22}$ &  8 & $10^{-15}$  &  68   \\ 
$2^{15}$   &    188258 &  638  & 64 & $2^{23}$ &  8 & $10^{-16}$  &  69   \\ 
$2^{16}$   &    376143 &  678  & 64 & $2^{24}$ &  8 & $10^{-17}$  &  77   \\ 
$2^{17}$   &    751781 &  719  & 64 & $2^{25}$ &  8 & $10^{-18}$  &  85   \\ 
$2^{18}$   &   1505016 &  765  & 64 & $2^{26}$ &  8 & $10^{-19}$  &  93   \\ 
$2^{19}$   &   3015969 &  808  & 32 & $2^{26}$ &  8 & $10^{-19}$  & 110   \\ 
$2^{20}$   &   6082313 &  864  & 64 & $2^{26}$ & 16 & $10^{-46}$  & 175   \\ \bottomrule

	\end{tabular}
\caption{Generation time in nanoseconds per 64-bit value using Horner's scheme, Gao-Mateer's FFT and an implementation of our constant-time generator}
\label{tab:experimentalresults}
\end{table}

\section*{Acknowledgment}
We are grateful to Martin Dietzfelbinger who gave feedback on an early version of the paper, allowing us to significantly enhance the presentation.

\bibliographystyle{IEEEtran}
\bibliography{focs}

\begin{thebibliography}{10}
\providecommand{\url}[1]{#1}
\csname url@samestyle\endcsname
\providecommand{\newblock}{\relax}
\providecommand{\bibinfo}[2]{#2}
\providecommand{\BIBentrySTDinterwordspacing}{\spaceskip=0pt\relax}
\providecommand{\BIBentryALTinterwordstretchfactor}{4}
\providecommand{\BIBentryALTinterwordspacing}{\spaceskip=\fontdimen2\font plus
\BIBentryALTinterwordstretchfactor\fontdimen3\font minus
  \fontdimen4\font\relax}
\providecommand{\BIBforeignlanguage}[2]{{%
\expandafter\ifx\csname l@#1\endcsname\relax
\typeout{** WARNING: IEEEtran.bst: No hyphenation pattern has been}%
\typeout{** loaded for the language `#1'. Using the pattern for}%
\typeout{** the default language instead.}%
\else
\language=\csname l@#1\endcsname
\fi
#2}}
\providecommand{\BIBdecl}{\relax}
\BIBdecl

\bibitem{hagerup1998}
T.~Hagerup, ``Sorting and searching on the word {RAM},'' in \emph{Proc.
  {STACS}'98}, 1998, pp. 366--398.

\bibitem{siegel2004}
A.~Siegel, ``On universal classes of extremely random constant-time hash
  functions,'' \emph{{SIAM} J. Comput.}, vol.~33, no.~3, pp. 505--543, 2004.

\bibitem{schmidt1995}
J.~P. Schmidt, A.~Siegel, and A.~Srinivasan, ``Chernoff-{H}oeffding bounds for
  applications with limited independence,'' \emph{{SIAM} J. Discrete Math.},
  vol.~8, no.~2, pp. 223--250, 1995.

\bibitem{alon2008}
N.~Alon and N.~Asaf, ``k-wise independent random graphs,'' in \emph{Proc.
  {FOCS}'08}, 2008, pp. 813--822.

\bibitem{braverman2010}
M.~Braverman, ``Polylogarithmic independence fools {$AC^{0}$} circuits,''
  \emph{J. {ACM}}, vol.~57, no.~5, pp. 28:1--28:10, 2010.

\bibitem{chor1985}
B.~Chor, O.~Goldreich, J.~Hastad, J.~Freidmann, S.~Rudich, and R.~Smolensky,
  ``The bit extration problem or t-resilient functions,'' in \emph{Proc.
  {FOCS}'85}, 1985, pp. 396--407.

\bibitem{alon1992}
N.~Alon, O.~Goldreich, J.~H{\aa}stad, and R.~Peralta, ``Simple constructions of
  almost k-wise independent random variables,'' \emph{Random Structures \&
  Algorithms}, vol.~3, no.~3, pp. 289--304, 1992.

\bibitem{naor1993}
J.~Naor and M.~Naor, ``Small-bias probability spaces: efficient constructions
  and applications,'' \emph{{SIAM} J. Comput.}, vol.~22, no.~4, pp. 838--856,
  1993.

\bibitem{goldreich2010}
O.~Goldreich, \emph{A primer on pseudorandom generators}.\hskip 1em plus 0.5em
  minus 0.4em\relax Providence, {RI}: American Math. Soc., 2010, vol.~55.

\bibitem{nisan1992}
N.~Nisan, ``Pseudorandom generators for space-bounded computation,''
  \emph{Combinatorica}, vol.~12, no.~4, pp. 449--461, 1992.

\bibitem{joffe1974}
A.~Joffe, ``On a set of almost deterministic k-independent random variables,''
  \emph{Ann. Prob.}, vol.~2, no.~1, pp. 161--162, 1974.

\bibitem{wegman1981}
J.~L. Carter and M.~N. Wegman, ``New hash functions and their use in
  authentication and set equality,'' \emph{J. Comput. System Sci.}, vol.~22,
  no.~3, pp. 265--279, 1981.

\bibitem{gathen2013}
J.~von~zur Gathen and J.~Gerhard, \emph{Modern computer algebra}, 3rd~ed.\hskip
  1em plus 0.5em minus 0.4em\relax Cambridge, {UK}: Cambridge University Press,
  2013.

\bibitem{bostan2005}
A.~Bostan and {\'E}.~Schost, ``Polynomial evaluation and interpolation on
  special sets of points,'' \emph{J. Complexity}, vol.~21, no.~4, pp. 420--446,
  2005.

\bibitem{granlund1994}
T.~Grandlund and P.~L. Montgomery, ``Division by invariant integers using
  multiplication,'' \emph{{SIGPLAN} Not.}, vol.~29, no.~6, pp. 61--72, 1994.

\bibitem{dietzfelbinger2012}
M.~Dietzfelbinger, ``On randomness in hash functions (invited talk),'' in
  \emph{Proc. {STACS}'12}, 2012, pp. 25--28.

\bibitem{thorup2013}
M.~Thorup, ``Simple tabulation, fast expanders, double tabulation, and high
  independence,'' in \emph{Proc. {FOCS}'13}, 2013, pp. 90--99.

\bibitem{larsen2012}
K.~G. Larsen, ``Higher cell probe lower bounds for evaluating polynomials,'' in
  \emph{Proc. {FOCS}'12}, 2012, pp. 293--301.

\bibitem{kedlaya2008}
K.~S. Kedlaya and C.~Umans, ``Fast modular composition in any characteristic,''
  in \emph{Proc. {FOCS}'08}, 2008, pp. 146--155.

\bibitem{shoup2009}
V.~Shoup, \emph{A computational introduction to number theory and algebra},
  2nd~ed.\hskip 1em plus 0.5em minus 0.4em\relax Cambridge, {UK}: Cambridge
  University Press, 2009.

\bibitem{capalbo2002}
M.~Capalbo, O.~Reingold, S.~Vadhan, and A.~Wigderson, ``Randomness conductors
  and constant-degree lossless expanders,'' in \emph{Proc. {STOC}'02}, 2002,
  pp. 659--668.

\bibitem{tashma2007}
A.~{Ta-Shma}, C.~Umans, and D.~Zuckerman, ``Lossless condensers, unbalanced
  expanders, and extractors,'' \emph{Combinatorica}, vol.~27, no.~2, pp.
  213--240, 2007.

\bibitem{capalbo2005}
M.~Capalbo, ``Explicit bounded-degree unique-neighbor concentrators,''
  \emph{Combinatorica}, vol.~25, no.~4, pp. 379--391, 2005.

\bibitem{bhattacharya2004}
M.~Bhattacharya, R.~Creutzburg, and J.~Astola, ``Some historical notes on the
  number theoretic transform,'' in \emph{Proc. 2004 Int. {TICS} Workshop on
  Spectral Methods and Multirate Signal Processing}, 2004.

\bibitem{duhamel1990}
P.~Duhamel and M.~Vetterli, ``Fast {Fourier} transforms: a tutorial review and
  state of the art,'' \emph{Signal Processing}, vol.~19, no.~4, pp. 259--299,
  1990.

\bibitem{preparata1977}
F.~P. Preparata and D.~V. Sarwate, ``Computational complexity of {Fourier}
  transforms over finite fields,'' \emph{Mathematics of Computation}, vol.~31,
  no. 131, pp. 740--751, 1977.

\bibitem{gueron2014}
S.~Gueron and M.~E. Kounavis, ``Intel carry-less multiplication instruction and
  its usage for computing the {CGM} mode,'' Intel Corporation, Tech. Rep.,
  2014.

\bibitem{seroussi1998}
G.~Seroussi, ``Table of low-weight binary irreducible polynomials,''
  Hewlett-Packard Laboratories, Tech. Rep., 1998.

\bibitem{pagh2008}
A.~Pagh and R.~Pagh, ``Uniform hashing in constant time and optimal space,''
  \emph{{SIAM} J. Comput.}, vol.~38, no.~1, pp. 85--96, 2008.

\bibitem{dietzfelbinger2003}
M.~Dietzfelbinger and P.~Woelfel, ``Almost random graphs with simple hash
  functions,'' in \emph{Proc. {STOC}'03}, 2003, pp. 629--638.

\bibitem{gao2010}
S.~Gao and T.~Mateer, ``Additive fast {Fourier} transforms over finite
  fields,'' \emph{{IEEE} Trans. Inf. Theory}, vol.~56, no.~12, pp. 6265--6272,
  2010.

\bibitem{mitzenmacher2005}
M.~Mitzenmacher and E.~Upfal, \emph{Probability and computing}.\hskip 1em plus
  0.5em minus 0.4em\relax New York, {NY}: Cambridge University Press, 2005.

\bibitem{mitzenmacher2014}
M.~Mitzenmacher, R.~Pagh, and N.~Pham, ``Efficient estimation for high
  similarities using odd sketches,'' in \emph{Proc. {WWW}'14}, 2014, pp.
  109--118.

\end{thebibliography}
\end{document}